\documentclass[12 pt]{amsart}
\usepackage{graphicx}

\usepackage{amsthm,amsmath,latexsym,amssymb,amsfonts,amscd,hyperref}
\usepackage{color,tikz,tikz-cd}
\usepackage[all]{xy}
\newtheorem{theorem}{Theorem}[section]
\newtheorem{lemma}[theorem]{Lemma}
\newtheorem{proposition}[theorem]{Proposition}

\newtheorem{cor}[theorem]{Corollary}
\theoremstyle{definition}
\newtheorem{definition}[theorem]{Definition}

\theoremstyle{remark}

\numberwithin{equation}{section}

\newcommand{\ww}{{\mathcal{W}}}

\begin{document}
%\phantom{a}
%\vspace{-1.5cm}
\title[Wick type deformation quantization of contact metric manifolds]{Wick type deformation quantization of contact metric manifolds}

\author{Boris M. Elfimov}
\address{Physics Faculty, Tomsk State University, Lenin ave. 36, Tomsk 634050, Russia}
\email{e1fimov@mail.ru}
\author{ Alexey A. Sharapov }
\address{Physics Faculty, Tomsk State University, Lenin ave. 36, Tomsk 634050, Russia}
\email{sharapov@phys.tsu.ru}

\thanks {The work of B.M.E. was supported 
by  the Foundation for the Advancement of Theoretical Physics and Mathematics ``BASIS''. The work of A.A.Sh. was supported by the FAPESP Grant 2022/13596-8.}



\begin{abstract}
We construct a Wick-type deformation quantization of contact metric manifolds. The construction is fully canonical and involves no arbitrary choice. 
Unlike the case of symplectic or Poisson manifolds, not every classical observable on a general contact metric manifold can be promoted to a quantum one due to possible obstructions to quantization. We prove, however, that all these obstructions disappear for Sasakian manifolds. 
\end{abstract}

\maketitle

\section{Introduction}

In the recent paper \cite{elfimov2022deformation}, we formulated the problem of deformation quantization for contact manifolds  and solved it using a modification of Fedosov's method \cite{Fedosov:1996}.   The main novelty compared to the case of symplectic or Poisson manifolds is  that  not every classical observable -- a Reeb-invariant function on a contact manifold -- admits  quantization. As a result, the algebra of quantum observables may well be smaller than the algebra of classical ones. The study of obstructions to quantization yielded  some new invariants of contact manifolds. 

One can regard the above quantization  as a contact counterpart of the Weyl-type deformation quantization of symplectic manifolds. 
This is the most natural quantization scheme if no extra geometric structure is assumed on symplectic or contact manifolds. In many interesting situations, however, symplectic manifolds come equipped with (integrable) Lagrangian polarizations, which are important, for example, for their  physical interpretation as phase spaces of mechanical systems. An invariant Lagrangian polarization is also an integral part of the geometric quantization method. In the most favorable case,  there is a pair of transverse Lagrangian polarizations, which  may be complex and mutually conjugate. The last situation is 
best exemplified  by (almost) K\"ahler manifolds with their holomorphic and antiholomorphic polarizations. By definition, a Wick quantization of a K\"ahler manifold is a quantization adapted to the underlying holomorphic structure, see Sec. \ref{S5}. There is a large body of literature devoted to various aspects of and approaches to the problem, starting with the pioneering works of Berezin \cite{Berezin1975GeneralCO, Berezin_1975}.  In Fedosov's approach, the Wick deformation quantization of K\"ahler manifolds was in-depth studied in \cite{bordemann1997fedosov, tamarkin1998fedosov, Dolgushev_2001, karabegov2001almost, la2016infinite}.

The odd-dimensional counterpart of almost K\"ahler geometry is known as contact metric geometry. This can be defined as the geometry of contact manifolds endowed with compatible Riemannian metrics. Incorporating a metric structure in the fabric of contact geometry  allows one to polarize a (maximally nonintegrable) contact distribution into holomorphic and antiholomorphic parts. In this paper, we use this polarization 
to formulate  a Wick-type deformation quantization of contact metric manifolds. 
As in \cite{elfimov2022deformation}, our construction largely relies on Fedosov's method. The rich geometry of contact metric manifolds allows for several types of integrability  conditions to be imposed. These lead to three special contact metric geometries: K-contact, CR, and Sasakian; the last one is the intersection of the first two. All these integrability conditions play an important part in the Wick deformation quantization developed below. In particular, we show that no obstructions arise for the Wick deformation quantization of Sasakian manifolds and the property of being K-contact ensures the absence of the first obstruction to quantization.  

To conclude this introduction, we  list the basic differences between the Weyl- and Wick-type deformation quantizations of contact manifolds. First, the contact metric manifolds enjoy a natural contact connection, making the Wick quantization completely canonical. The absence of a canonical connection on general contact manifolds introduces an inevitable arbitrariness in the Weyl quantization.  Second, the obstructions to the Weyl quantization of classical observables start from the second order in the deformation parameter, while in the Wick quantization, they arise already in the first order. Third, the differential operators that control the leading obstructions to quantization are completely different:  they are third-order for Weyl quantization and first-order for Wick.

\section{Contact metric geometry}\label{Sec2}

We follow the definitions and conventions of the classical  papers \cite{Sasaki1960ONDM, SH, Sasaki1962H}. They differ slightly from those adopted in more recent expositions of contact metric geometry \cite{Blair2002, BG}. The difference is explained in \cite[Thm. 7.5]{sasaki1965almost}. 

One can regard contact manifolds as odd-dimensional analogues of symplectic ones. Recall that a $1$-form $\lambda$ on a ($2m+1$)-dimensional manifold $M$ is called {\it contact} if 
\begin{equation}\label{int}
    \lambda \wedge (d\lambda)^m \neq 0
\end{equation}
at each point of $M$. The pair $(M,\lambda)$ is called a  {\it contact manifold}. Geometrically, Eq. (\ref{int}) means that the {\it contact distribution} $\mathcal{C}=\ker \lambda $ on $M$ is `maximally nonintegrable'. As a consequence, each contact manifold carries the canonical volume form  $v=\lambda \wedge (d\lambda)^m$ and an exact presymplectic structure defined by the $2$-form $\omega =d\lambda$ of rank $2m$. 
The kernel of the presymplectic form $\omega$, called {\it characteristic distribution}, is generated by the {\it Reeb vector field} $\xi$ which is chosen to satisfy 
\begin{equation}
 i_\xi\omega=0\,,\qquad i_\xi\lambda=1\,.   
\end{equation}
This gives the canonical splitting 
$
    TM=\mathcal{C}\oplus L_\xi
$,
where $L_\xi$ is the trivial line bundle generated by $\xi$.

It is known that every contact manifold admits a Riemannian metric $g$ and a tensor field $\phi$ of type $(1,1)$ that satisfy the following relations: 
\begin{equation}\label{algrel}
    \begin{array}{l}
      \lambda(X)=g(X,\xi) \,, \qquad  \omega (X,Y)=g(X,\phi(Y)) \,,\\[3mm]
       g(\phi(X), \phi(Y))=g(X,Y)-\lambda(X)\lambda(Y)\,,   \qquad \phi^2=-I+\xi\otimes \lambda\,,\\[3mm]
       \phi (\xi)=0\,, \qquad \lambda \circ \phi=0\,.
    \end{array}
\end{equation}
Here, $X$ and $Y$ are arbitrary vector fields on $M$ and we treat the $(1,1)$-tensor field\footnote{In \cite{Sasaki1960ONDM}, the tensor $\phi$ was dubbed a {\it fundamental singular collineation}, but this name  did not stick.} $\phi$ as defining an endomorphism of the tangent bundle $TM$. One usually calls $g$ an {\it associated metric}.   It is known that the associated metric is far from being unique and not all of Rels. (\ref{algrel}) are independent: the last two, for example, follow from the others.  We refer to the data $(\phi, \xi,\lambda, g)$ as a {\it contact metric structure} and call $M$, endowed with such a structure, a {\it contact metric manifold}. 

Every contact metric manifold carries a canonical {\it Jacobi structure} \cite{kirillov31local, lichnerowicz1978varietes}. This is given by the Reeb vector field $\xi$ and the bivector $\pi$ defined by the rule
\begin{equation}\label{pi}
    (\alpha\wedge \beta)(\pi)=\bar g (\alpha, \phi^t(\beta))\,.
\end{equation}
Here, $\alpha$ and $\beta$ are arbitrary $1$-forms; $\bar g$ and $\phi^t$ stand for the inverse metric and the transpose of the endomorphism $\phi$. One can see that 
\begin{equation}
    [\pi,\pi]_{SN}=2\pi\wedge \xi\,,\qquad [\xi,\pi]_{SN}=0\,.
\end{equation}
The subscript SN refers to the Schouten--Nijenhuis bracket on poly-vector fields. 
Alternatively, using the aforementioned volume form $v$, one can define the same bivector $\pi$ as dual to the form  $\lambda\wedge\omega^{m-1}$. This shows that the Jacobi structure $(\xi,\pi)$ is completely determined by the contact structure alone \cite[Sec. 6]{Sasaki1962H}.  
Given a Jacobi structure, one can make the space of smooth functions on $M$ into a Lie algebra for the following bracket:
\begin{equation}\label{bracket}
\{f,g\}=(df\wedge dg)(\pi)+f dg(\xi) -g df(\xi)\,.
\end{equation}
This is called the {\it Jacobi bracket} of functions $f,g\in C^\infty(M)$. 
It should be noted that  (\ref{bracket}) is not a Poisson bracket as it fails to satisfy the Leibniz identity:
$$
\{f,gh\}-\{f,g\}h-g\{f,h\}=gh df(\xi)\,.
$$
Notice, however, that the Jacobi bracket gives rise to a Poisson algebra structure on the commutative subalgebra of $\xi$-invariant functions $C^\infty_\xi(M)\subset C^\infty(M)$.  

{\it By a deformation quantization of a contact manifold $(M,\lambda)$ we mean a deformation quantization of the Poisson algebra $\big(C^\infty_\xi(M), \{\,\cdot\,,\,\cdot\,\}\big)$.} A precise definition will be given in the next section.

In order to develop the differential geometry of contact metric manifolds one usually introduces one or another adapted connection. By definition, a {\it contact metric connection} $\nabla$ is an affine connection  on $M$ that parallelizes the contact metric structure:
\begin{equation}\label{cmc}
    \nabla \lambda=0\,,\quad \nabla\xi=0, \quad \nabla \omega=0\,,\quad \nabla g=0\,,\quad \nabla \phi=0\,.
\end{equation}
The definition is somewhat redundant since the first three relations follow from the last two  \cite[Lem. 4.1]{STADTMULLER20122170}.   
It is known that a contact metric connection always exists and can be constructed as follows \cite[Thm. 11]{SH}. Let $\nabla^g$ denote the Levi-Civita connection that is compatible with the metric $g$. In each local coordinate system, define a $(1,2)$-tensor $S$ by the expression
\begin{equation}\label{S}
     S^i_{jk} = -\frac12 \phi^i_l \nabla^g_k\phi^l_j  - \frac12 \lambda_j \nabla^g_k\xi^i + \xi^i \nabla^g_k \lambda_j \,.
\end{equation}
Using (\ref{algrel}) one can readily verify that the new affine connection $\nabla=\nabla^g+S$ meets all conditions (\ref{cmc}). As usual, the commutator of covariant derivatives evaluated on a vector field $v$,
\begin{equation}\label{tc}
    [\nabla{}_k, \nabla{}_l] v^i = R^i_{jkl} v^j + T^j_{kl} \nabla_j v^i\,,
\end{equation}
gives the curvature  $R^i_{jkl}$ and the torsion tensor $T^i_{jk} = S^i_{jk} - S^i_{kj}$. The latter is always nonzero. Indeed, it follows from the first identity in (\ref{cmc}) that 
\begin{equation}
 \omega_{ij}= \lambda_kT_{ji}^k  \,.
\end{equation}
The tensor $S$ is known as the {\it potential} of a metric connection $\nabla=\nabla^g+S$. The following easily verified identity shows that the torsion of a metric connection contains the same information  as its potential:
\begin{equation}
    S_{ij}^k=\frac12T^k_{ij}+\frac12(g_{im}T^m_{jn}+g_{jm}T^m_{in})g^{nk}\,.
\end{equation}
Let us also introduce the covariant torsion and curvature tensors 
\begin{equation}\label{TR}
    T_{kij}=\omega_{kn}T^n_{ij}\,,\qquad R_{klij}=\omega_{kn}R^n_{lij}\,.
\end{equation}
They enjoy the following symmetry properties:
\begin{equation}\label{urel}
\qquad R_{klij}=R_{lkij}=-R_{klji}\,,\qquad T_{kij}+T_{ijk}+T_{jki}=0\,.
    \end{equation}

Let  $\hat \phi$ denote the restriction of the endomorphism $\phi: TM\rightarrow TM$ onto the contact distribution $\mathcal{C}\subset TM$. Then $\hat \phi^2=-I$; and hence, $\phi$ induces an almost complex structure on $\mathcal{C}$. Therefore, the complexification of the vector bundle $\mathcal{C}$ splits into the two transverse and  mutually conjugate subbundles composed of $\pm i$-eigenspaces of the operator $\hat\phi$:
\begin{equation}\label{PC}
\begin{array}{c}
    \mathcal{C}_{\mathbb{C}}=\mathcal{P}\oplus \bar {\mathcal{P}}\,, \qquad \mathcal{P}\cap \bar {\mathcal{P}}=0\,,   \\[3mm]
    \mathcal{P}=\{ X-i\hat \phi X \;|\; \forall X\in \mathcal{C} \}\,,\qquad  \bar{\mathcal{P}}=\{ X+i\hat \phi X \;|\; \forall X\in \mathcal{C} \}\,.   
    \end{array}
    \end{equation}
This splitting equips $M$ with an {\it almost CR structure} of hyper surface type, which can be viewed as an odd-dimensional counterpart of the almost complex structure\footnote{The abbreviation CR means either `Cauchy--Riemann' or `Complex-Real'. }. Furthermore, the presence of the positive definite Levi form 
\begin{equation}
    L_\lambda(X,Y)=\omega(\phi (X),Y)=g(X,Y)\qquad \forall X,Y\in \mathcal{C}\,,
\end{equation}
allows one to identify this almost CR structure as {\it strongly pseudoconvex} \cite{DT}. An almost CR structure is called integrable (or just CR structure) if the holomorphic subbundle $\mathcal{P}$ is involutive, i.e.,  $[\mathcal{P}, \mathcal{P}]\subset \mathcal{P}$. The last condition amounts to the vanishing of the $(1,2)$-tensor field 
\begin{equation}\label{Ntensor}
    Q_{ij}^k=\nabla^g_j\phi_i^k+ \xi^k\phi_i^n\nabla^g_j \lambda_n+\phi^k_n\nabla^g_j\xi^n\lambda_i\,.
\end{equation}
This result is due to Tanno \cite{tanno1989variational}. 

For the purposes of deformation quantization, we also introduce the dual vector bundles $\mathcal{C}^\ast $ and $\mathcal{P}^\ast$. The sections of the former are naturally identified with  the $1$-forms on $M$ that annihilate the Reeb vector field $\xi$.  To characterise the latter it is convenient to equip the complexified cotangent bundle $T^\ast_{\mathbb{C}}M$ with the Hermitian form 
\begin{equation}\label{h}
    h(\alpha,\beta)=\bar g(\alpha, \beta +i\phi^t(\beta))=i\pi(\alpha,\beta )+\bar g(\alpha,\beta)\,,
\end{equation}
where   $ \alpha$ and  $\beta$ are arbitrary $1$-forms.  Let $\hat h$ denote the restriction of $h$ onto $\mathcal{C}^\ast$. Then  $\mathcal{P}^\ast\subset \mathcal{C}^\ast_{\mathbb{C}}$
coincides with the right kernel of $\hat{h}$, that is, $\hat h(-,\mathcal{P}^\ast)=0$. In view of Hermiticity  the antiholomorphic distribution $\bar{\mathcal P}^\ast$ generates then the left kernel of $\hat h$. The contact metric connection (\ref{cmc}) clearly respects $h$. We will refer to the Hermitian form (\ref{h}) as the {\it Wick tensor}. 

Finally,  a contact metric structure $(\phi, \xi, \lambda,  g,)$ is said to be {\it $K$-contact} if $\xi$ is a Killing vector of the associated metric $g$, i.e., $\mathcal{L}_\xi g=0$. If, in addition, the corresponding CR structure is integrable ($Q=0$), then the contact metric manifold is called {\it Sasakian}. It is known    \cite[Thm. 11]{Sasaki1962H}   that  the property of being  Sasakian is equivalent to the identity
\begin{equation}\label{Sas}
    2\nabla^g_k \omega_{ij} = \lambda_ig_{kj}-\lambda_jg_{ki}\,.
\end{equation}
In dimension $3$, the holomorphic distribution $\mathcal{P}$, being one-dimensional, is automatically integrable. Therefore every $3$-dimensional K-contact manifold is Sasakian.

\section{Wick deformation quantization of contact metric manifolds}\label{Sec3}
As mentioned in the previous section, by the deformation quantization of a contact metric manifold $M$ we understand the quantization of the Poisson algebra $(C^\infty_\xi(M), \{\,\cdot \,,\,\cdot\,\})$ associated with the canonical Jacobi structure on $M$. It is remarkable that the Jacobi 
bracket $\{\,\cdot \,,\,\cdot\,\}$ is determined by the contact structure alone and in no way involves the associated metric.   The general definition of the Weyl deformation quantization of contact manifolds was formulated in our recent paper \cite{elfimov2022deformation}. 
The presence of the canonical Hermitian form (\ref{h}) on $M$ suggests the following definition of  Wick deformation quantization.

%Let us start with a couple of basic definitions.  
\begin{definition}\label{D41} Let $\nu$ be a  formal deformation parameter and $C_\mathbb{C}^\infty(M)$ be the commutative algebra of smooth complex-valued functions on $M$. 
By a {\it Wick deformation quantization} of a contact metric manifold $(M;  \phi, \xi, \lambda, g)$ we understand  a pair $(\ast, \Delta )$, where $\ast$ and  $\Delta$ are, respectively, bilinear and linear operators on   $C_{\mathbb{C}}^\infty(M)[\![\nu]\!]$ of the form
$$ a\ast b= ab+\frac\nu 2h(da,  db)+\sum_{k=2}^\infty \nu^kC_k(a,b)\,,\qquad
    \Delta a=\xi a+\sum_{k=1}^\infty \nu^k\Delta_k a\,.$$
It is also assumed that 
\begin{enumerate}
    \item $C_k$ and $\Delta_k$ are  differential operators on $C^\infty(M)$, which are extended to $C_{\mathbb{C}}^\infty(M)[\![\nu]\!]$ by $\mathbb{C}[\![\nu]\!]$-linearity;
    \item  the elements of the subspace $\ker \Delta$ form an associative algebra w.r.t. the $\ast$-product, that is, 
$$
\Delta(a\ast b)=0\quad \mbox{and}\quad  (a\ast b)\ast c=a\ast(b\ast c)
$$
whenever $\Delta a=\Delta b=\Delta c=0$. 
\end{enumerate}
The pair $(\ker \Delta, \ast)$ is called the {\it algebra of quantum observables}. 
\end{definition}

An important property of the operators $\ast$ and $\Delta$ is their {\it locality}, i.e., the possibility of restriction on functions 
defined in any open domain $U\subset M$. By abuse of language, we will refer to this as the restriction of the operators $\ast$ and $\Delta$  onto the domain $U$. 

In order to construct the Wick deformation quantization  above we will follow Fedosov's approach \cite{Fedosov:1996}. Let $U$ be a coordinate chart on $M$ with coordinates $\{x^i\}$ and denote by $\{y^i\}$ the fibre coordinates on the tangent bundle $TU$ w.r.t. the natural frame $\{\partial/\partial x^i\}$. For each point $x\in M$ define the {\it formal Wick algebra} $W_x$ as the space spanned by the formal power series with complex coefficients
\begin{equation}
    a(\nu,y)=\sum_{k,l=0}^\infty \nu^k a_{k,i_1\cdots i_l }y^{i_1}\cdots y^{i_l}
\end{equation}
endowed with the $\circ$-product 
\begin{equation}
    (a\circ b)(y)=\exp\left(\frac{\nu}{2}h^{ij}(x)\frac{\partial^2}{\partial y^i\partial z^j}\right) a(y)b(z)\Big|_{z=y}\,,
\end{equation}
$h^{ij}(x)$ being the components of the Wick tensor at $x$. The $\circ$-product is known to be associative, but not commutative. Taking the union of algebras $W_x$, we obtain the bundle of formal Wick algebras $W$, whose space of sections will be denoted by $\mathcal{W}$. Prescribing the formal variables $y$'s and $\nu$ the degrees $1$ and $2$, respectively,  makes $\mathcal{W}$ into an $\mathbb{N}$-graded algebra, that is, $\mathcal{W}=\bigoplus_{n\geq 0}\mathcal{W}_n$.  We  refer to this grading as the {\it $\nu$-degree}. 
Associated with the $\nu$-degree is the descending filtration 
$$
{F}^k\mathcal{W}\supset {F}^{k+1}\mathcal{W}\,,\qquad {F}^k\mathcal{W}=\bigoplus_{n\geq k}\mathcal{W}_{n}\,,\qquad k=0,1,\ldots\,,
$$
which defines topology and convergence in the space of formal power series. 
Sometimes we will use another natural $\mathbb{N}$-grading on the  vector space $\ww$, which is given by the polynomial degree in $y$'s. We  refer to it as {\it $y$-degree} and write $\ww=\bigoplus_{n\geq 0} \ww_{(n)}$. Unlike the $\nu$-degree, the $y$-degree does not respect the $\circ$-product. 

The most interesting to us will be the subalgebra $\mathcal{W}^\xi\subset\mathcal{W}$ of $\xi$-transverse sections of $W$. This is defined as 
\begin{equation}
    \mathcal{W}^\xi\ni a\quad \Longleftrightarrow \quad \xi^i\frac{\partial a}{\partial y^i}=0\,.
    \end{equation}
Tensoring the algebra $\mathcal{W}$ with the exterior algebra $\Lambda(M)=\bigoplus_{n\geq 0} \Lambda_n(M)$ of differential forms on $M$, we obtain the bigraded algebra $\mathcal{W}\otimes \Lambda$, the second grading being the form degree. We will use the same  symbol $\circ$ for multiplication in the tensor product algebra. 
Let us also introduce the commutator of two elements of $\mathcal{W}\otimes \Lambda$ by 
\begin{equation}
    [a,b]=a\circ b-(-1)^{nm}b\circ a
\end{equation}
for all $a\in \mathcal{W}\otimes \Lambda_n$ and $b \in \mathcal{W}\otimes \Lambda_m$. The commutator makes the space $\mathcal{W}\otimes \Lambda$ into a graded Lie algebra w.r.t. the form degree. It is easy to see that the centre of this Lie algebra is given by $\Lambda(M) [\![\nu]\!]$.  The algebra $\mathcal{W}\otimes \Lambda$ is endowed with the differential $\delta$ of bidegree $(-1,1)$:
\begin{equation}
    \delta = dx^i \wedge \frac{\partial}{\partial y^i} \,, \qquad \delta^2  = 0\,.
\end{equation}
It is clear that
\begin{equation}
    \delta(a \circ b) = \delta a  \circ b + (-1)^n a \circ \delta b  \qquad \forall a \in \mathcal{W}\otimes \Lambda_n\,, \quad \forall b \in \mathcal{W}\otimes \Lambda\,.
\end{equation}
Upon restriction onto the differential subalgebra $\mathcal{W}^\xi\otimes \Lambda$ one can write the differential $\delta$ as
\begin{equation}
    \delta a = \frac{i}{\nu} [\omega_{ij} y^i dx^j, a]\qquad \forall a\in \mathcal{W}^\xi\otimes \Lambda\,.
\end{equation}
In order to describe the cohomology of the differential $\delta$ it is convenient to split $\Lambda(M)$ into the direct sum
\begin{equation}
    \Lambda = \Lambda^\xi \oplus \Lambda^\lambda \,, \qquad \Lambda^\xi = \{a \in\Lambda \,|\, i_\xi a = 0\}\,, \qquad \Lambda^\lambda = \{\lambda a\, |\, \forall a \in \Lambda\}\,.
\end{equation}
Since the differential $\delta$ leaves invariant the subspaces $\Lambda^\xi$ and $\Lambda^\lambda$, the groups of $\delta$-cohomology
\begin{equation*}
    H_{m,n}(\mathcal{W}\otimes \Lambda) = \mathrm{ker}(\delta : \mathcal{W}_m\otimes \Lambda_n\rightarrow \mathcal{W}_{m-1}\otimes \Lambda_{n+1})/\delta (\mathcal{W}_{m+1}\otimes \Lambda_{n-1})
\end{equation*}
split as $H_{m,n}(\mathcal{W}\otimes \Lambda) = H_{m,n}(\mathcal{W}\otimes \Lambda^\xi) \oplus H_{m,n}(\mathcal{W}\otimes \Lambda^\lambda)$.

\begin{proposition}\label{cohomology}
With the above notation 
$$
    \begin{array}{l}
        H_{0,0}(\mathcal{W}^\xi\otimes \Lambda) \simeq H_{0,0}(\mathcal{W}^\xi\otimes \Lambda^\xi)\simeq C_{\mathbb{C}}^\infty (M) [\![\nu]\!]\,,\\[3mm]
        H_{0,1}(\mathcal{W}^\xi\otimes \Lambda)\simeq H_{0,1}(\mathcal{W}^\xi\otimes \Lambda^\lambda)\simeq C_{\mathbb{C}}^\infty (M) [\![\nu]\!]\,,\\  [3mm]           
        H_{m,n}(\mathcal{W}^\xi\otimes \Lambda) = 0\,, \quad \forall m > 0 \quad \mbox{or}\quad n > 1\,.       
    \end{array}
$$
 
\end{proposition} 
For the sake of completeness we reproduce the proof of this statement from \cite{elfimov2022deformation}.

\begin{proof}  The $(1,1)$-tensor $P^i_j = \delta^i_j - \xi^i\lambda_j$ defines an orthogonal projection  in the tangent bundle $TM$: 
\begin{equation*}
    P^2 = P\,, \qquad \mathrm{Tr}P = 2m\,,
\end{equation*}
\begin{equation*}
    P^i_j = \pi^{ik} \omega_{kj}\,, \quad \lambda_i P^i_j =0\,,\quad P^i_j \xi^j = 0\,, \quad P^k_i \omega_{kj} = \omega_{ij} \,, \quad \pi^{ik} P^j_k = \pi^{ij}\,.
\end{equation*}
 Define the operator $\delta^\ast : \mathcal{W}^\xi_m\otimes \Lambda_n \rightarrow \mathcal{W}^\xi_{m+1}\otimes \Lambda_{n-1}$ as
    \begin{equation*}
        \delta^\ast = y^i P_i^j i_{\frac{\partial}{\partial x^j}}\,.
    \end{equation*}
    It is easy to see that $(\delta^\ast)^2 = 0$, $\delta^\ast \lambda = 0$, and
    \begin{equation}
        (\delta \delta^\ast + \delta^\ast \delta) a = (m+n) a - \lambda i_\xi a  \qquad \forall a \in \mathcal{W}^\xi_{m}\otimes \Lambda_{n}\,.
    \end{equation}
    Notice that $\delta^\ast$ is not a derivation of the $\circ$-product, but it maps $\xi$-transverse forms of $\ww^\xi\otimes \Lambda$ to $\xi$-transverse. Now define the operator 
    \begin{equation}
        \delta^{-1}: \mathcal{W}^\xi_m\otimes \Lambda_n \rightarrow \mathcal{W}^\xi_{m+1}\otimes\Lambda_{n-1}    
        \end{equation}
        by the relations
    \begin{equation}
        \delta^{-1} a' = \left\{
            \begin{array}{cc}
                \frac{1}{m+n} \delta^\ast a' \,, & m+n >0\,; \\
                0 \,, & m = n =0\,;
            \end{array} \right. \qquad
        \delta^{-1} a'' = \left\{
            \begin{array}{cc}
                \frac{1}{m+n-1}\delta^\ast a''\,, & m+n > 1\,;  \\
                0\,, & m+n = 1 \,;
            \end{array} \right.
    \end{equation}
    for all $a'\in \mathcal{W}_{(m)}\otimes \Lambda_n^\xi$ and $a'' \in \mathcal{W}_{(m)}\otimes \Lambda^\lambda_n$. Writing $a=a'+a''$, one can check that
    \begin{equation}\label{HDR}
        (\delta \delta^{-1} + \delta^{-1} \delta)a = a - \Pi a \,,
    \end{equation}
    where the operator $\Pi$ acts as
    \begin{equation*}
        \Pi a = a(x, 0, 0, \nu) + \lambda (i_\xi a) (x, 0, 0, \nu) \qquad \forall \, a = a(x, y, dx, \nu) \in \mathcal{W}^\xi\otimes \Lambda\,.
    \end{equation*}
    Clearly, $\Pi^2=\Pi$.    Eq. (\ref{HDR}) implies that all nontrivial $\delta$-cocycles are nested in the space $\mathrm{ker}(1-\Pi)$ spanned by the sums $a + \lambda b$ with $a, b \in C_{\mathbb{C}}^\infty (M)[\![\nu]\!]$. On the other hand, the elements of the form $a + \lambda b$ are obviously nontrivial cocycles and the proposition follows. 
    \end{proof}

Every  contact metric connection $\nabla$ on the tangent bundle $TM$ induces a linear connection on the associated bundle of formal Wick algebras $W$. The latter, in turn, gives rise to an exterior covariant derivative on $\mathcal{W}^\xi\otimes \Lambda$. In each coordinate chart, it is defined as
\begin{equation}\label{d}
d_\nabla a =dx^i\wedge \nabla_i a=dx^i\wedge \left(\frac{\partial a}{\partial x^i}-y^j\Gamma_{ji}^k\frac{\partial a}{\partial  y^k}\right)\qquad \forall a\in \ww^\xi\otimes\Lambda\,,
\end{equation}
$\Gamma_{ij}^k = \Gamma_{ij}^k(x)$ being the coefficients of the contact metric connection. Since $\nabla$ respects the Wick tensor $h$, the graded Leibniz rule holds
\begin{equation}\label{L}
d_\nabla (a\circ b)=d_\nabla a\circ b+(-1)^n a\circ d_\nabla b\qquad \forall a\in\ww\otimes \Lambda_n\,,  \quad \forall b\in\ww\otimes \Lambda\,.
\end{equation}
Define the elements 
\begin{equation}
T = \frac12 T_{kij} y^k dx^i \wedge dx^j \in \ww^\xi_1\otimes \Lambda_2 \quad  \mbox{and}\quad  R = \frac14  R_{ijkl} y^i y^j dx^k \wedge dx^l \in \ww^\xi_2\otimes \Lambda_2
\end{equation}
associated with the torsion and curvature of the contact metric connection (\ref{TR}). With the help of Rels. (\ref{urel}) and the identity 
\begin{equation}
    g^{in}R_{nkl}^j+ g^{jn}R_{nkl}^i=0
    \end{equation}
one can find
\begin{equation}
    (d_\nabla \delta + \delta d_\nabla)a= -\frac{i}{\nu} [T, a]\,,\qquad d_\nabla^2a=\frac i\nu[R,a]\qquad \forall a\in \ww^\xi\otimes \Lambda\,.    
    \end{equation}
The Jacobi identity for the double commutators of the operators $\delta$ and $d_\nabla$ implies the following Bianchi identities:
\begin{equation}
    \delta T=0\,,\qquad \delta R=d_\nabla T\,,\qquad d_\nabla R=0\,.
\end{equation}

Following Fedosov, we consider more general connections in the Wick bundle $W$, which are defined by exterior covariant derivatives  of the form
\begin{equation}\label{fc}
    Da = d_\nabla a - \delta a+ \frac{i}{\nu} [r, a] = d_\nabla a+ \frac{i}{\nu}[r-\omega_{ij} y^i dx^j  , a] 
\end{equation}
for some $r\in F^2\ww^\xi\otimes \Lambda_1$. A straightforward calculation yields 
\begin{equation}
    D^2 a = \frac{i}{\nu}[\Omega, a] \,, 
\end{equation}
where
\begin{equation}\label{Omega}
    \Omega = \omega - \delta r + T+R + d_\nabla r + \frac{i}{\nu} r \circ r \in \ww^\xi\otimes \Lambda_2\,, \qquad \omega=\frac12\omega_{ij}dx^i\wedge dx^j\,.
\end{equation}
The two-form $\Omega$ is called the \textit{Weyl curvature} of the Fedosov connection $D$. 
%It should be noted that the shift of $r$ by a central one-form does not affect both the connection $D$ and the curvature $\Omega$. Therefore only the de Rham %cohomology $[\Omega] \in H^2(M) [\![\hslash]\!]$ has an invariant meaning.  
According to the Bianchi identity $D\Omega=0$. A  connection $D$ is called {\it abelian} if $D^2=0$. 
This implies that the curvature $\Omega$ of an abelian connection is a closed two-form of $\Lambda_2(M)[\![\nu]\!]$.

\begin{theorem}
    For  a given contact metric connection $\nabla$ there exists a unique abelian connection (\ref{fc}) with  Weyl curvature $\Omega=\omega$ and $\delta^{-1}r =0\,.$
\end{theorem}

The theorem states the existence of a unique solution to the equations
\begin{equation}\label{ab}
\delta r = T +R + d_\nabla r + \frac{i}{\nu} r \circ r \,,\qquad \delta^{-1}r=0\,.
\end{equation}
The proof of this fact essentially repeats that of \cite{elfimov2022deformation}, and we omit it here. In practical terms, the desired solution is obtained by solving the equation
 \begin{equation}\label{EqForr}
        r = \delta^{-1}(T + R) + \delta^{-1}\left( d_\nabla r + \frac{i}{\nu} r \circ r\right)\,,
    \end{equation}
which is an obvious  consequence of  (\ref{ab}) and 
(\ref{HDR}). Since $\delta^{-1}(R + T) \in \mathcal{W}^\xi_2\otimes \Lambda_2$ and the operator $\delta^{-1}$ increases the $\nu$-degree by one unit, we can solve Eq. (\ref{EqForr}) by iteration. It also follows from (\ref{EqForr}) that $\delta^{-1} r = 0$. 
Iterating Eq. (\ref{EqForr}) yields 
\begin{equation}\label{r}
        r = \frac13 y^i y^mP^j_m T_{ijk}   P^k_ndx^n + \frac{1}{2}y^i y^j T_{ijk}\xi^k  \lambda+\cdots\,,
\end{equation}
where dots stand for elements of $F^3\ww^\xi\otimes \Lambda_1$. The higher-degree terms involve the curvature tensor.

Associated with an abelian connection $D$ is the {\it algebra of  flat sections} $\ww_D$. By definition, 
$$
\ww_D=\{a\in \ww^\xi\;|\;Da=0\}\,.
$$
Since $D$ differentiates the $\circ$-product, the flat sections form a subalgebra in $\ww^\xi$.  The algebra $\ww_D$ is the central object of the Fedosov quantization. To clarify the structure of the space $\ww_D$ we need the following lemma of \cite{elfimov2022deformation}.

\begin{lemma}
    For any $a \in \mathcal{W}^\xi$ the equation
    \begin{equation}\label{b-a}
        b = a + \delta^{-1}(D + \delta)b 
    \end{equation}
    has a unique solution for $b \in \mathcal{W}^\xi$. Furthermore, $\Pi b = 
    \Pi a$.
\end{lemma}
The statement follows from the observation that the linear operator $\delta^{-1}(D + \delta)$ raises the $\nu$-degree by one unit. Therefore one can solve Eq. (\ref{b-a}) by iterations to get a unique solution for $b$. Since the image of $\delta^{-1}$ depends on positive powers of $y$'s, applying the projection  $\Pi$ to both sides of (\ref{b-a}) gives $\Pi b=\Pi a$. The assignment $a\mapsto b$ defines an injective  $\mathbb{C}[\![\nu]\!]$-linear map
\begin{equation}\label{Q}
    Q: C_{\mathbb{C}}^\infty (M)[\![\nu]\!] \rightarrow \mathcal{W}^\xi\,,
\end{equation}
called a \textit{quantum lift}. Using (\ref{r}) and iterating (\ref{b-a}) for $a\in C^\infty_{\mathbb{C}}(M)[\![\nu]\!]$, one can find that up to the second order in $\nu$-degree
\begin{equation}\label{a}
        Qa  =  \displaystyle  a +  y^jP^i_j\nabla_i a   + \frac12 y^ny^m P^i_nP^j_m\nabla_i \nabla_j a  +\frac16 y^i y^nP^j_n T_{ijk}  \pi^{km} \nabla_{m} a +\cdots
\end{equation}
Again, the higher-order terms involve the curvature tensor and its covariant derivatives. 

The composition of maps
 \begin{equation}\label{PDQ}
     \Delta=i_\xi \Pi D Q
     \end{equation}
     defines an operator $\Delta: C_{\mathbb{C}}^\infty(M)[\![\nu]\!]\rightarrow C_{\mathbb{C}}^\infty(M)[\![\nu]\!]$.   
Using Eqs. (\ref{r}) and (\ref{a}), one can find that
     \begin{equation}\label{deltaa}
        \Delta  = \xi -\frac{\nu}{8} T_{ijk}\xi^k H^{ij|nm}\nabla_n\nabla_m +\frac{\nu}{24}T_{ijk}T_{nm p}\xi^pH^{ij|nm}\pi^{kl}\nabla_l+\mathcal{O}(\nu^2)\,,
\end{equation}
where 
\begin{equation}
\begin{array}{c}
    H^{ij|nm}=\pi^{in}g^{jm}+\pi^{jm}g^{in}+\pi^{jn}g^{im}+\pi^{im}g^{jn}\,,\\[3mm]
    H^{ij|nm}=H^{ji|nm}=H^{ij|mn}=-H^{nm|ij}\,.
    \end{array}
\end{equation}

The next theorem is completely analogous to Theorem 4.4  of Ref. \cite{elfimov2022deformation}. 
\begin{theorem}\label{T35}
    Each element $a\in \ww_D$ is completely determined by its projection $\Pi a\in C_{\mathbb{C}}^\infty(M)[\![\nu]\!]$ through  the formula 
\begin{equation}
    a=Q\Pi a\,.
\end{equation}
    An element $a_0\in C_{\mathbb{C}}^\infty(M)[\![\nu]\!]$ is the projection of a flat section if and only if $\Delta a_0=0$.  
\end{theorem}

As an immediate corollary of this theorem, we conclude that $\ww_D\simeq \ker \Delta$. The quantum lift (\ref{Q}) allows one to pull back the $\circ$-product from $\ww^\xi$ to $C_\mathbb{C}^\infty [\![\nu ]\!]$. This results in the following $\ast$-product: 
\begin{equation}\label{sp}
    a\ast b=\Pi(Qa \circ Qb)\qquad \forall a,b \in C^\infty_{\mathbb{C}}(M)[\![\nu]\!] \,.
\end{equation}
Unlike the $\circ$-product, the $\ast$-product is not associative. Notice, however, that it becomes associative upon restriction to the subspace $\ker \Delta$. By construction, the associative algebra $(\ker \Delta, \ast)$ is isomorphic to  the algebra of flat sections $(\ww_D, \circ)$.

Now, it is straightforward to check that the linear and bilinear operators (\ref{PDQ}) and (\ref{sp}) do satisfy all the  properties assumed by Definition \ref{D41}. Hence, they define a Wick deformation quantization on a contact metric manifold $(M; \phi, \xi, \lambda, g)$. Notice that $1\in \ww_D$ and the algebra of quantum observables  $(\ker \Delta, \ast)$ is unital.

\section{Quantum observables}
In this section, we will examine more closely the space of quantum observables $\ker \Delta$. In particular, we will consider the implications of a contact metric manifold being K-contact.  

To study the solution space of the equation $\Delta a=0$, we expand the operator $\Delta$ and an element $a\in C^\infty_{\mathbb{C}}(M)[\![\nu]\!]$ in  
powers of $\nu$: 
\begin{equation}
    \Delta=\xi+\nu\Delta_1+\nu^2\Delta_2+\cdots\,,\qquad a=a_0+\nu a_1+\nu^2 a_2+\cdots \,.
\end{equation}
On substituting these back into the equation, we get an infinite sequence of relations, which start as
\begin{equation}\label{cheq}
   \xi a_0=0\,,\qquad \Delta_1 a_0+\xi a_1=0\,, \quad \ldots
\end{equation}
The first equation says that $a_0$ is a $\xi$-invariant function on $M$. 
%Such functions are usually called  {\it basic}. 
As mentioned in Sec. \ref{Sec2}, the $\xi$-invariant functions form a closed Poisson algebra $C^\infty_\xi(M)$ w.r.t. the Jacobi bracket (\ref{bracket}). It is natural to call 
$\big (C^\infty_\xi(M), \{\cdot\, ,\,\cdot \}\big )$ the {\it algebra of classical observables}. We thus conclude that the $\nu$-expansion of any quantum observable starts from a classical one. The higher powers of $\nu$ are `quantum corrections' we need to add to a 
classical observable $a_0$ to promote it at the quantum level.  The corrections are to be found successively from the chain of equations (\ref{cheq}).  
The solvability of these equations is not ensured in advance. It depends on a particular contact metric structure and a classical observable $a_0$ to be quantized.  Therefore, in contrast to the deformation quantization of Poisson manifolds, not every classical observable on a contact manifold admits quantization.
We say that a classical observable $a_0\in C^\infty_\xi(M)$ is {\it quantizable} if system (\ref{cheq}) admits a solution. Otherwise, we speak about obstructions to quantization. 
The first obstruction comes with the second equation in (\ref{cheq}). It says that the function $-\Delta_1 a_0$, defined by a classical observable $a_0$, must be a gradient of some other function $a_1\in C^\infty(M)$ along the Reeb vector field $\xi$. 
    Recall that $\Delta_1$ is a second-order differential operator defined by Eq. (\ref{deltaa}). After lengthy, albeit straightforward calculations, one can bring this operator into the form     
\begin{equation}\label{d1}
    \Delta_1=-\frac14(\mathcal{L}_\xi g^{ij})\nabla_i\nabla_j +\frac{1}{12}(\mathcal{L}_\xi g^{ij})(\nabla^g_i\omega_{jk})\pi^{kn}\nabla_n\,.
\end{equation}
The following proposition is now obvious.
\begin{proposition}
        On every K-contact manifold, $\Delta_1=0$.
        \end{proposition}
 Thus, every classical observable on a K-contact manifold can be promoted to a quantum one up to the first order in $\nu$ (e.g. setting $a_1=0$). Using the physics language, we can say that every classical observable on a K-contact manifold is quantizable in semi-classical approximation. 

 In the general case, the operator $\Delta_1$ is different from zero, so the equation $\xi a_1+\Delta_1 a_0=0$ may present a real  obstruction to quantization. 
To better understand this obstruction it is useful to rewrite  the operator (\ref{d1}) in terms of the metric connection:
\begin{equation}
      \Delta_1 =-\frac14(\mathcal{L}_\xi g^{ij})\nabla^g_i\nabla^g_j  -\frac14(\mathcal{L}_\xi g^{ij} ) S_{ij}^k\nabla^g_k    +\frac{1}{12}(\mathcal{L}_\xi g^{ij})(\nabla^g_i\omega_{jk})\pi^{kn}\nabla^g_n\,.\end{equation}
(The operator is supposed to act on scalar functions.) Next, one can find that
\begin{equation}
    (\mathcal{L}_\xi g^{ij} ) S_{ij}^k\nabla^g_k =\frac12(\mathcal{L}_\xi g^{ij})(\nabla^g_i\omega_{jk})\pi^{kn}\nabla^g_n +\frac12(\mathcal{L}_\xi g^{ij})(\mathcal{L}_\xi g_{ij})\nabla_\xi^g\,.
    \end{equation}
    The standard formulas for the commutator of the Lie and covariant derivatives,
    \begin{equation}
        [\mathcal{L}_\xi, \nabla_i]f=0\,,\qquad [\mathcal{L}_\xi, \nabla_i^g]\alpha_j= -({}^g\!R^k_{jni}\xi^n  +\nabla^g_i\nabla^g_j\xi^k)\alpha_k\,,
    \end{equation}
    allows us to write
\begin{equation}
    (\mathcal{L}_\xi g^{ij})\nabla^g_i\nabla^g_j=[\mathcal{L}_\xi, \mathbb{L}]- g^{ij}[\mathcal{L}_\xi,\nabla^g_i]\nabla^g_j=[\mathcal{L}_\xi, \mathbb{L}]   + ({}^g\!R^k_{nlm}g^{nm}\xi^l+g^{nm}\nabla_n^g\nabla_m^g \xi^k)\nabla_k^g\,,
    \end{equation}
    where ${}^g\!R_{nlm}^k$ is the Riemann curvature tensor and
    \begin{equation}
         \mathbb{L}=g^{nm}\nabla_n^g\nabla_m^g
    \end{equation}
is the Laplace operator on scalar functions. Combining all the parts together, we finally get
\begin{equation}\label{4d}
    -4\Delta_1= [\mathcal{L}_\xi, \mathbb{L}] +\frac12 (\mathcal{L}_\xi g_{ij})(\mathcal{L}_\xi g^{ij})\nabla_\xi^g+\nabla_\zeta\,.
\end{equation}
Here we introduced the  vector field
\begin{equation}
    \zeta=\Big({}^g\!R^i{}_{nkm}g^{nm}  \xi^k +  g^{nm}\nabla^g_n\nabla^g_m  \xi^i   +\frac{1}{6}(\mathcal{L}_\xi g^{nm})(\nabla^g_n\omega_{mk})\pi^{ki}\Big)\frac{\partial}{\partial x^i}\,.
\end{equation}
The commutator in (\ref{4d}) suggests to look for the first quantum correction in the form 
\begin{equation}
    a_1=\frac14 \mathbb{L}a_0 +\frac14 a_1'\,,
\end{equation}
where $a'_1$  satisfies the {first-order}  differential equation 
\begin{equation}\label{xiz}
    \xi a_1'=\zeta a_0\,.
\end{equation}

Below, we state some necessary conditions for this equation to have a solution. Recall that integral curves of the Reeb vector field $\xi$ are called {\it characteristics} of a contact metric manifold $M$. The function $a_0$, being $\xi$-invariant, must be constant on each characteristic.
The most interesting are closed characteristics, that is, periodic orbits of the Reeb flow\footnote{The famous Weinstein's  conjecture states that  every compact manifold $M$ enjoys at least one closed characteristic, see  \cite[Ch. 3.4]{Blair2002}.}.   Let us  multiply both sides of (\ref{xiz}) by $\lambda$ and then integrate over a closed characteristic $\gamma$. As the left-hand side vanishes for an obvious reason, we get the  condition
 \begin{equation}
    \Psi_\gamma[a_0]:=\int_\gamma \lambda \zeta a_0=0\,.
 \end{equation}
We call $\Psi_\gamma[a_0]$ the {\it characteristic of a classical observable} $a_0$ (associated with a closed characteristic $\gamma\subset M$). A nonzero number $\Psi_\gamma[a_0]$ represents thus the first obstruction to the quantization of $a_0\in C^\infty_\xi(M)$.  

Suppose now that the contact metric manifold $M$ is compact. Then we can endow the space of smooth functions $C^\infty(M)$ with a positive-definite inner product
\begin{equation}
    (a,b)=\int_M v{a}b\,.
\end{equation}
Here $v$ is the canonical volume form on $M$. Let $f'$ denote the derivative of a smooth function $f\in C^\infty(\mathbb{R})$. Multiplying both sides of Eq. (\ref{xiz}) by $f'(a_0)$ and using the $\xi$-invariance of $a_0$, we get
\begin{equation}
    \xi (f'(a_0) a_1)=\zeta f(a_0)\,.
\end{equation}
Since the volume form is $\xi$-invariant, integrating the last equality over $M$ yields 
\begin{equation}\label{Phif}
    \Phi_f[a_0]:=\int_M v\zeta f(a_0)=0\,.
\end{equation}
 We call this functional the {\it $f$-character of a classical observable  $a_0$}. Again, a nonzero value of $\Phi_f[a_0]$ implies an obstruction to quantization. Define the {\it character of  a contact metric manifold} $M$ by the formula
 \begin{equation}
\chi :=\mathrm{div}_v \zeta=\nabla_i^g\zeta^i\,.
 \end{equation}
 The equality takes place because the canonical and Riemannian volume forms on $M$ coincide up to a constant factor \cite[Thm. 6]{sasaki1965almost}. 
 Integrating by parts in (\ref{Phif}), we get
 \begin{equation}\label{char}
     \Phi_f[a_0]=-(\chi, f(a_0))\,.
     \end{equation}
Eqs. (\ref{Phif}, \ref{char}) say that any smooth function of a quantizable classical observable $a_0$ must be orthogonal to the character of $M$. It would be interesting to clarify the geometric properties of the vector field $\zeta$ and implications of the conditions $\zeta=0$ or $\chi=0$.

\section{Sasakian manifolds and Wick property}\label{S5}
Now is the time to explain the adjective `Wick' in the name of the star product above. First, we recall the definition of the Wick deformation quantization on complex manifolds. Let $M$ be a complex manifold endowed with a $\ast$-product, that is, a formal associative deformation of the pointwise multiplication of complex-valued functions on $M$. Since the deformation is assumed to be local, one can restrict the $\ast$-product to
any open domain $U\subset M$. Then the $\ast$-product is said to be of {\it Wick type}\footnote{Another commonly used term is a {\it deformation quantization with separation of variables}  \cite{karabegov1996deformation}.} if for any holomorphic function $a$ and any antiholomorphic function $b$ on $U$ the following equalities hold:
\begin{equation}
    c\ast a=c\cdot a\,,\qquad b\ast c= b\cdot c\,.
\end{equation}
Here $c$ is an arbitrary smooth function on $U$.  Below, we will establish similar properties for the $\ast$-product (\ref{sp}) on Sasakian manifolds.   
This result is quite expected, given the analogy between K\"ahler and Sasakian geometries. Less obvious is the fact that any classical observable on a Sasakian manifold is quantizable.

Let us start with some auxiliary constructions and statements. First, we define the following two subspaces in the $\circ$-product algebra $\ww^\xi\otimes \Lambda$:
\begin{equation}
\begin{array}{l}
\mathrm{ L}=\big \{a\in\ww^\xi\otimes \Lambda \quad\big | \quad (b\circ a)|_{y=0}=0\,,\quad \forall b\in \ww^\xi\otimes \Lambda \big\}\,,\\[3mm]
     \mathrm{R}=\big \{a\in\ww^\xi\otimes \Lambda \quad\big | \quad (a\circ b)|_{y=0}=0\,,\quad \forall b\in \ww^\xi\otimes \Lambda \big\}\,.\end{array}
\end{equation}
Clearly,  $\mathrm L$ is a left ideal and $\mathrm R$ is a right ideal in $\ww^\xi\otimes \Lambda$. Then the intersection 
$\mathrm A=\mathrm L\cap\mathrm R$ is a subalgebra in $\ww^\xi\otimes \Lambda$, while the subspaces $\mathrm L$ and $\mathrm R$ are bi-modules over $\mathrm A$: 
\begin{equation}
 \mathrm A\circ \mathrm A\subset \mathrm A\,,\qquad   \mathrm  A\circ\mathrm  L\circ \mathrm A\subset \mathrm L\,,\qquad \mathrm A\circ \mathrm R\circ \mathrm A\subset \mathrm R\,.
 \end{equation}

\begin{lemma}\label{L41}
    The subspaces $\mathrm L$, $\mathrm R$, and $\mathrm A$ are invariant under the action of the operators $d_\nabla$ and $\delta^{-1}$. 
\end{lemma}
\begin{proof}
    Let $a\in \mathrm L$ and $b\in \ww^\xi\otimes \Lambda_n$.  Using the definition of the covariant differential (\ref{d}) and the Leibniz rule (\ref{L}),  we can write
    \begin{equation}
    \begin{array}{c}
( b\circ d_\nabla a)|_{y=0}=(-1)^nd_\nabla (b\circ a)|_{y=0}-(-1)^n (d_\nabla b\circ a)|_{y=0}\\[3mm]
         =(-1)^n\left.\Big( d(b\circ a)-dx^i \wedge y^j\Gamma_{ji}^k \frac{\partial}{\partial y^k}(b\circ a)\Big)\right |_{y=0}=(-1)^n d\Big ((b\circ a)|_{y=0}\Big)  =0\,. \end{array}  
    \end{equation}
     In order to prove the invariance of $\mathrm L$ under $\delta^{-1}$ it is enough to show that $y^i\cdot a\in L$ for all $a\in L$. We start with the obvious identity
     \begin{equation}
         y^i\cdot a=y^i\circ a+X^i a\,,
     \end{equation}
     where 
     \begin{equation}
         X^i=-\frac{\nu}{2}h^{ij}\frac{\partial}{\partial y^j}\,.
     \end{equation}
Note that the operators $X^i$ differentiate the $\circ$-product. Therefore, we may proceed as follows: 
\begin{equation}
\begin{array}{rcl}
 b\circ (y^i\cdot a)|_{y=0}&=&(b\circ y^i)\circ a|_{y=0}+b\circ X^ia|_{y=0} \\[3mm]
     &=&  X^i(b\circ a)|_{y=0} - (X^i b)\circ a|_{y=0}\\[3mm]&=&y^i \cdot (b\circ a)|_{y=0}- (y^i\circ b)\circ a|_{y=0}=0\,.
     \end{array}
\end{equation}
This proves\footnote{In fact, we proved more. Since $\ww^\xi\otimes\Lambda$, viewed as graded commutative  algebra w.r.t. the dot product, is generated by the $y$'s,  $b\cdot a\in \mathrm L$ for all $a\in \mathrm L$ and $b\in \ww^\xi\otimes\Lambda$.} that $\mathrm L$ is invariant under the action of $\delta^{-1}$. The invariance of $\mathrm R$ is proved in the same way. Taken together, this means the invariance of the algebra $\mathrm A=\mathrm L\cap \mathrm R$. 
\end{proof}

\begin{lemma}\label{L42}
    On every contact metric manifold a contact metric curvature $R\in \mathrm A$. 
\end{lemma}
\begin{proof}
    Using Rels. (\ref{algrel}), (\ref{pi}), and (\ref{h}), one can easily verify that
    \begin{equation}
        \omega_{ij}h^{in}h^{jm}=0\,,\qquad h^{ni}h^{mj}\omega_{ij}=0\,.
    \end{equation}
    The form $h$ being Hermitian, the equalities are obtained from each other by complex conjugation. By the definition of a contact metric connection,   
    \begin{equation}
       0= [\nabla_k,\nabla_l]h^{ij}=R^i_{nkl}h^{nj}+R^j_{nkl}h^{in}  \,.         
        \end{equation}
        As a consequence,
        \begin{equation}
            \omega_{pi}(R^i_{nkl}h^{nj}+R^j_{nkl}h^{in}  ) h^{pm}=R_{pnkl}  h^{pm}h^{nj}+R^j_{nkl}(\omega_{pi}  h^{pm}h^{in} )
            =R_{pnkl}  h^{pm}h^{nj}=0\,.            
        \end{equation}
            On the other hand, 
            \begin{equation}
                (R\circ b)|_{y=0}=\frac 14 \left (\frac{\nu}{2}\right)^2R_{ijkl}dx^k\wedge dx^l h^{in}h^{jm}\left.\frac{\partial^2 b}{\partial y^n\partial y^m}\right|_{y=0}=0\,.
            \end{equation}
This means $R\in \mathrm R$. Similarly, one can see that $R\in \mathrm L$; and hence, $R\in \mathrm A=\mathrm{L}\cap \mathrm{R}$.
\end{proof}

\begin{lemma}\label{L43}
    On every Sasakian manifold $\delta^{-1}T\in \mathrm{A}$. 
\end{lemma}
\begin{proof}
Starting from definition (\ref{S}), we first compute the covariant potential tensor
\begin{equation}
    \begin{array}{rcl}
     S_{ijk}:= \omega_{in}S^n_{jk} &=&\displaystyle -\frac12\omega_{in}\phi^n_m\nabla^g_k\phi^m_j-\frac12\omega_{in}\lambda_j\nabla_k^g\xi^n \\ [3mm] 
   &=& \displaystyle\frac12(g_{in}-\lambda_i\lambda_n)\nabla^g_k\phi^n_j+\frac12 g_{nm}\phi^m_i\lambda_j\nabla^g_k\xi^n \\[3mm]&=&\displaystyle 
   \frac12\nabla^g_k\omega_{ij}-\frac12\lambda_i\lambda_n\nabla^g_k\phi^n_j+\frac12 \phi^n_i\lambda_j\nabla_k^g\lambda_n\\[3mm]
   &=&\displaystyle 
   \frac12\nabla^g_k\omega_{ij}+\frac12\lambda_i\phi^n_j\nabla^g_k\lambda_n+\frac12\lambda_j\phi^n_i\nabla_k^g\lambda_n\,.   
    \end{array}
\end{equation}
Here we used the algebraic identities $\omega_{ik}\phi^k_j=-g_{ij}+\lambda_i\lambda_j$ and $\phi_i^j\lambda_j=0$. Since every Sasakian manifold is automatically K-contact,
\begin{equation}
    \nabla^g_i\lambda_j=\frac12\omega_{ij}
\end{equation}
and we can write
\begin{equation}
    \begin{array}{rcl}
        S_{ijk}&=&\displaystyle \frac12\nabla^g_k\omega_{ij}+\frac14\lambda_i\phi^n_j\omega_{kn}+\frac14\lambda_j\phi^n_i\omega_{kn} \\[3mm]    
        &=&\displaystyle \frac12\nabla^g_k\omega_{ij}+\frac14(\lambda_i\lambda_j\lambda_k -\lambda_i g_{jk}-\lambda_jg_{ik})\,.   
    \end{array}
\end{equation}
Combining the last formula with the identity 
\begin{equation}
    \nabla^g_i\omega_{jk}+\nabla^g_j\omega_{ki}+\nabla^g_k\omega_{ij}=0\,,
    \end{equation}
we get the following expression for the covariant torsion tensor on $K$-contact manifolds:
\begin{equation}\label{TS}
    T_{ijk}=S_{ijk}-S_{ikj}=\frac12\nabla_i^g\omega_{kj}-\frac14(\lambda_jg_{ki}-\lambda_kg_{ji})\,.
\end{equation}
Using the defining property of Sasakian manifolds (\ref{Sas}), we can simplify the last expression to 
\begin{equation}
    T_{ijk}=\frac12(\lambda_kg_{ij}-\lambda_jg_{ik})\,.
\end{equation}
Hence,
\begin{equation}
    T=\frac12 y^i T_{ijk}dx^j\wedge dx^k=\frac12 y^ig_{ij}dx^j\wedge \lambda\quad \mbox{and}\quad \delta^{-1}T=\frac14y^iy^j(g_{ij}-\lambda_i\lambda_j)\lambda\,.
\end{equation}
Notice that
\begin{equation}
 g_{ij}h^{in}h^{jm}=h^{ni}h^{mj}g_{ij}=0\,, \qquad \lambda_i\lambda_j h^{in}h^{jm}=h^{ni}h^{mj}\lambda_i\lambda_j=\xi^n\xi^m\,.
    \end{equation}
Now we are ready to show that $\delta^{-1}T\in \mathrm{R}$. Indeed, 
\begin{equation}
    (\delta^{-1}T\circ b)|_{y=0}=\frac14\left(\frac{\nu}{2}\right)^2 \lambda (g_{ij}-\lambda_i\lambda_j)h^{in}h^{jm}\left.\frac{\partial^2 b}{\partial y^n\partial y^m}\right|_{y=0}=0\end{equation}
for all $b\in \ww^\xi\otimes \Lambda$. In the same way, one can see that $\delta^{-1}T\in \mathrm{L}$. Therefore, $\delta^{-1}T\in \mathrm{A}$. 

\end{proof}

As mentioned in Sec. \ref{Sec2}, every Sasakian manifold carries a CR structure. This is given by  the integrable  (anti)holomorphic distribution (\ref{PC}).    Let $U$  be an open domain in $M$. A complex-valued function $a\in C_{\mathbb{C}}^\infty(U)$ is said to be {\it CR-holomorphic} (resp. {\it CR-antiholomorphic}) if it is annihilated by all vector fields belonging to the antiholomorphic distribution $\bar{\mathcal{P}}|_U$ (resp. $\mathcal P|_U$). By $\mathbb{C}[\![\nu]\!]$-linearity, these definition extends to the space  $C_{\mathbb{C}}^\infty(U)[\![\nu]\!]$. Of particular interest to us will be CR-(anti)holomorphic functions that are $\xi$-invariant. It is easy to see that any $\xi$-invariant CR-holomorphic (resp. -antiholomorphic) on $U$ function $a$ satisfies the equations 
\begin{equation}
    h^{ji}\frac{\partial a}{\partial x^i}\Big|_U=0\qquad \Big(\mbox{resp.}\quad \frac{\partial a}{\partial x^i}h^{ij}\Big|_U=0\Big)\,.
\end{equation}

\begin{theorem} Let $U$ be an open domain in a Sasakian manifold $M$. Consider the restriction of $\ast$ and $\Delta$ on $U\subset M$. Then any $\xi$-invariant CR-holomorphic (resp. -antiholomorphic) function  $a$ on $U$ is a quantum observable, i.e., $\Delta a=0$. Moreover, the Wick property holds
\begin{equation}\label{WP}
    c\ast a=c\cdot a\qquad (\mbox{resp.} \quad a\ast c=a\cdot c )
\end{equation}
for any quantum observable $c$ on $U$.

\end{theorem}

\begin{proof}
    The quantum lift  (\ref{Q}) of a CR-holomorphic function $a$ is determined by the Fedosov connection (\ref{fc}), where $r$ is found recursively from  Eq. (\ref{EqForr}). By Lemmas \ref{L41} -- \ref{L43}, $\delta^{-1}(T+R)\in \mathrm{A}$. Since $d_\nabla$ and $\delta^{-1}$ leave invariant the $\circ$-subalgebra $\mathrm{A}$, iteration of Eq. (\ref{EqForr}) gives a unique solution $r$ belonging to $\mathrm{A}$.  Let us expand both $r$ and  $Qa$ according to the $\nu$-degree:

\begin{equation}
    r=r_2+r_3+r_4+\cdots\,,\qquad Qa=a+a_1+a_2+\cdots\,.
\end{equation}
The expansion for $Qa$ follows from the recurrence relations 
\begin{equation}\label{aL}
    a_1=\delta^{-1}d_\nabla a\,, \qquad a_{n+1}=\delta^{-1}\Big(d_\nabla a_n +\frac{i}{\nu}\sum_{k=1}^n [r_{n-k+2}, a_k]\Big)\,,\quad n=1,2,\ldots
\end{equation}
We claim that $a_1\in \mathrm L$. Indeed, $a_1=y^i\nabla_i a$ and
\begin{equation}
    (c\circ a_1)|_{y=0}= \frac{\nu}{2}\frac{\partial c}{\partial y^j} h^{ji}\frac{\partial a}{\partial x^i}\Big|_{y=0}=0\qquad \forall c\in \ww^\xi\otimes \Lambda\,,
\end{equation}
as $a$ is $\xi$-invariant and CR-holomorphic on $U$.  Since $r_k\in \mathrm{A}$ and the $\mathrm A$-bimodule  $\mathrm{L}$ is invariant under the action of $d_\nabla$ and $\delta^{-1}$, Eq. (\ref{aL}) implies that all $a_k\in \mathrm{L}$. 

By Theorem \ref{T35}, $Qa$ is a flat section of $\ww_D$ iff
\begin{equation}
    \Delta a=i_\xi \Pi D Q a=(D_\xi Qa)|_{y=0}=0\,.
\end{equation}
Explicitly,
\begin{equation}
    \Delta a=\xi a-\sum_{k=1}^\infty \xi^i\frac{\partial a_k}{\partial y^i} \Big|_{y=0}+i_\xi \sum_{k=1}^\infty \Big(  d_\nabla a_k+\frac{i}{\nu}[r, a_k]\Big)\Big|_{y=0}=0\,.
\end{equation}
Here the first summand vanishes due to the $\xi$-invariance of $a$, the first sum is zero as all $a_k$'s are $\xi$-transversal, and 
the second sum goes to zero as an element of $\mathrm{L}$.     

It remains to prove the Wick property (\ref{WP}). Let $c$ be an arbitrary quantum observable, then
\begin{equation}
    c\ast a=\Pi (Qc\circ Qa)=c\cdot a+\sum_{k=1}^\infty ( Qc\circ a_k)|_{y=0}=c \cdot a
\end{equation}
because all $a_k\in \mathrm L$. The proof for CR-antiholomorphic functions is analogous. 

\end{proof}
It follows from the theorem above that for every CR-holomorphic function $a$ and every CR-antiholomorphic function $b$ on $U\subset M$, the product $b\ast a=b\cdot a$ is a quantum observable, that is, $\Delta (ab)=0$. On the other hand, the products of $\xi$-invariant CR-holomorphic and -antiholomorphic functions generate the whole space of $\xi$-invariant functions on $U$. This means that the operator $\Delta$, being local,  must be of the form $\Delta=B\xi$ for some formal differential operator $B=1+\mathcal{O}(\nu)$. Thus, we arrive at 

\begin{cor}
    On every Sasakian manifold $M$, all classical observables -- the elements of $C^\infty_\xi(M)$ -- are quantizable without any quantum corrections. 
\end{cor}

This remarkable property distinguishes Sasakian manifolds from other contact metric manifolds and conforms nicely to their intuitive understanding as odd-dimensional analogues of K\"ahler manifolds.

\end{document}